\documentclass[journal,twoside,web]{ieeecolor}

\usepackage{generic}
\usepackage{graphics} 
\usepackage{amsmath,amssymb,amsfonts,amsthm} 
\usepackage[font={small}]{caption}
\usepackage{bbold}
\usepackage{enumerate}
\let\labelindent\relax
\usepackage{enumitem}
\usepackage{mathtools}
\usepackage{dsfont}
\usepackage{cite}
\usepackage{apptools}

\usepackage[colorlinks = true,
            linkcolor = blue,
            urlcolor  = blue,
            citecolor = blue,
            anchorcolor = blue]{hyperref}
\usepackage{cleveref}

\theoremstyle{plain}
\newtheorem{theorem}{Theorem}[section]

\newtheorem{proposition}{Proposition}[section]
\newtheorem{lemma}{Lemma}[section]

\newtheorem{definition}{Definition}

\def\bs{\boldsymbol}
\def\mbb{\mathbb}

\def\bs{\boldsymbol}
\def\mbb{\mathbb}

\def\mcc{\mathcal{C}}
\def\mcx{\mathcal{X}}
\def\mcy{\mathcal{Y}}

\def\mcs{\mathcal{S}}
\def\mct{\mathcal{T}}
\def\dx{\delta_\mcx}
\def\dy{\delta_\mcy}

\def\costx{\kappa}
\def\costy{\sigma}

\newcommand{\ones}{\mathds{1}}

\def\bx{\boldsymbol{x}}
\def\btx{\textbf{x}}
\def\bty{\textbf{y}}
\def\by{\boldsymbol{y}}
\def\bu{\boldsymbol{u}}
\def\bX{\boldsymbol{X}}
\def\bY{\boldsymbol{Y}}

\def\bv{\boldsymbol{v}}

\def\E{\mathbb{E}}
\def\R{\mathbb{R}}
\def\F{\mathbb{F}}
\def\P{\mathbb{P}}
\def\Rp{\mathbb{R}_{\geq 0}}

\def\GL{\text{GL}}

\def\ML{\text{ML}}
\def\WL{\text{WL}}
\def\BS{\text{BS}}
\def\WC{\text{WC}}
\def\MLC{\text{ML-C}}
\def\MLM{\text{ML-M}}

\def\UB{\text{UB}}
\def\LB{\text{LB}}

\def\be{\begin{equation}}
\def\ee{\end{equation}}
\def\ba{\begin{aligned}}
\def\ea{\end{aligned}}
\def\ben{\begin{enumerate}}
\def\een{\end{enumerate}}
\def\bi{\begin{itemize}}
\def\ei{\end{itemize}}
\def\i{\item}
\def\v2{\vspace{2mm}}

\AtAppendix{\counterwithin{lemma}{subsection}}

\begin{document}

\title{Allocation of Heterogeneous Resources in General Lotto Games }

\author{Keith Paarporn,  Adel Aghajan,  Jason R. Marden 
\thanks{ K. Paarporn is with the Department of Computer Science at the University of Colorado, Colorado Springs. A. Aghajan is with the Department of Electrical and Computer Engineering at the University of California, San Diego. J. R. Marden is with the Department of Electrical and Computer Engineering at the University of California, Santa Barbara, CA. Contact: \texttt{kpaarpor@uccs.edu, adaghaja@ucsd.edu, jrmarden@ucsb.edu}. This work is supported by NSF grant \#ECCS-2346791, ONR grant \#N00014-20-1-2359, AFOSR grants \#FA9550-20-1-0054 and \#FA9550-21-1-0203, and the Army Research Lab through the ARL DCIST CRA \#W911NF-17-2-0181. This paper extends all results from a preliminary conference version \cite{aghajan2023equilibrium} to an arbitrary number of resource types. }
}

\maketitle

\begin{abstract}
    The allocation of resources plays an important role in the completion of system objectives and tasks, especially in the presence of strategic adversaries. Optimal allocation strategies are becoming increasingly more complex, given that multiple heterogeneous types of resources are at a system planner's disposal.  In this paper, we focus on deriving optimal strategies for the allocation of heterogeneous resources in a well-known competitive resource allocation model known as the General Lotto game. In standard formulations, outcomes are determined solely by the players' allocation strategies of a common, single type of resource across multiple contests. In particular, a player wins a contest if it sends more resources than the opponent. Here, we propose a multi-resource extension where the winner of a contest is now determined not only by the amount of resources allocated, but also by the composition of resource types that are allocated. We completely characterize the equilibrium payoffs and strategies for two distinct formulations. The first consists of a weakest-link/best-shot winning rule, and the second considers a winning rule based on a weighted linear combination of the allocated resources. We then consider a scenario where the resource types are costly to purchase, and derive the players' equilibrium investments in each of the resource types.
\end{abstract}

\section{Introduction}\label{sec:intro}

\begin{figure*} 
	\begin{center}
		\includegraphics[width=0.85\linewidth]{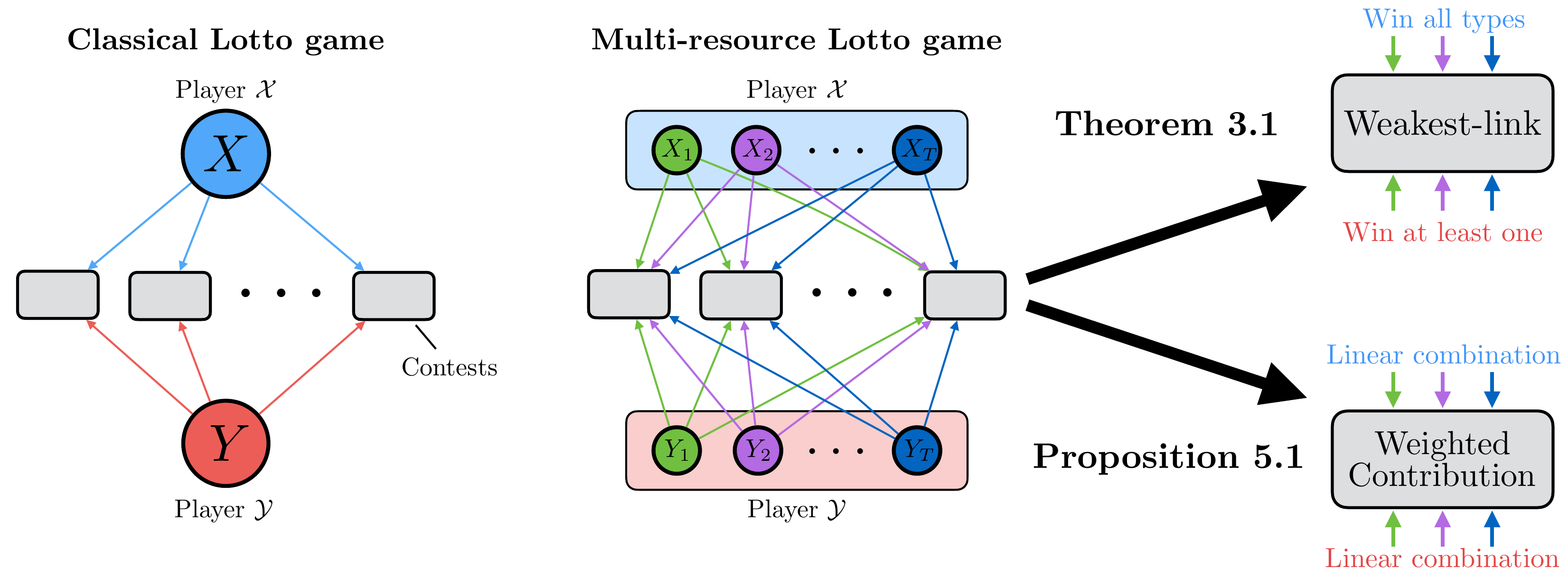} 
	\end{center}
	\caption{(Left) The classic General Lotto game, wherein a single resource type (e.g. money) is allocated to multiple simultaneous contests. Player $\mcx$ has a budget of $X \geq 0$ resources, and player $\mcy$ has a budget of $Y\geq 0$ resources. Success on a contest here simply depends on sending more resources than the opponent. (Right) The multi-resource General Lotto game. There are multiple resource types available to allocate (e.g. money, advertising, and human resources), where player $\mcx$ has budget $X_1\geq 0$ of type 1 resources, $X_2 \geq 0$ of type 2 resources, and so on. The success on each contest is now determined by a winning rule that depends on the combined allocation of resource types from both players. Our main contributions characterize equilibrium payoffs and strategies for two types of winning rules. Theorem \ref{thm:WL} considers the weakest-link rule, wherein player $\mcx$ needs to allocate more of every resource type to win a contest, whereas player $\mcy$ only needs to allocate more of only a single type of resource. Proposition \ref{thm:WC} considers a setting where each resource type has an associated weight, or effectiveness. A player wins the contest if the aggregate weighted amount of resources exceeds that of the opponent.}
	\label{fig:setup}
\end{figure*}

System planners are often responsible for allocating limited resources towards accomplishing multiple objectives, e.g. the allocation of autonomous agents to complete distributed tasks, or the allocation of security assets to prevent successful attacks against a network. The increasing complexity of these systems  require effective resource allocation strategies to become increasingly complex.  Part of this is the need to utilize multiple heterogeneous types of resources that are specialized for completing particular tasks. For example, multirobot systems require the allocation and coordination of agents with heterogeneous capabilities in order to complete multiple objectives \cite{Emam_2021,Notomista_2022,Banfi_2022}. Ensuring security in cyber-physical systems  requires the allocation of cybersecurity resources, physical security infrastructures, and human resources \cite{cardenas2009challenges,baheti2011cyber,mo2011cyber,jiang2018data,Ferdowsi_2020}. The successful completion of  objectives often cannot be accomplished through the allocation of only a single resource type. This poses new challenges as to how strategies for the allocation of heterogeneous resources should be conceived. Indeed, a system planner seeks to maximize performance by deploying multiple types of assets with different capabilities, effectiveness, and costs.

In this paper, we focus on deriving the optimal allocation of multiple, heterogeneous resources types in strategic adversarial environments. In particular, we propose an extended formulation of the General Lotto game, a popular model of competitive resource allocation between opponents \cite{Hart_2008,kovenock2021generalizations}. The General Lotto game is a variant of the original Colonel Blotto game, which models a strategic interaction where two budget-constrained players compete over a set of valuable contests, and the player that deploys more resources to a particular contest wins its associated value. In the standard formulations, the players have access only to a single common resource type, and the winner of each contest is determined purely by the amount of resources allocated to that contest -- this is known as the winner-take-all winning rule \cite{Kovenock_handbook_2012}. 

The primary literature on General Lotto and Colonel Blotto games almost exclusively considers the allocation of a single resource type \cite{Borel,Gross_1950,Roberson_2006,Golman_2009,Schwartz_2014,kovenock2021generalizations,Boix_2020}, e.g. only money, only UAVs, only troops, etc. From an applications standpoint, many important resource allocation problems cannot be addressed with a single-resource model. For example, consider a defender of a cyber network that needs to ensure security against malware, denial-of-service, and social engineering attacks, all of which require different types of resources to prevent \cite{longtchi2024internet}. This leads to more nuanced decision problems regarding the allocation of heterogeneous assets. In particular, the success of a competitor depends not only on the amount of resources allocated, but also on the \emph{composition} of resource types being deployed in contested environments. Consequently, the specification of new classes of multi-dimensional winning rules becomes an important modeling consideration.


In this paper, we consider formulations where players have access to and can allocate multiple resource types. Here, success on individual contests now depend on the multi-dimensional allocation of resource types from both players, and we consider two such winning rules. Figure \ref{fig:setup}  provides diagrams illustrating the differences between the standard single-resource formulations and our novel multi-resource formulation.


\subsection{Contributions}

We propose a framework to study classes of multi-resource General Lotto games, where players compete by allocating multiple types of resources. Our goal is to provide equilibrium characterizations for this new class of models. Specifically, we seek to analytically derive the equilibrium payoffs and strategies for both players. Our main contributions are as follows.
\begin{itemize}[leftmargin=*]
	\item The formulation of a novel multi-resource General Lotto game in Section \ref{sec:model}. We consider a finite number of resource types, and both players have limited budgets pertaining to each type. The winning rule on each contest is a function that depends on the allocation of all resource types to that contest.
	\item Our main contribution provides a complete equilibrium characterization in the case of a weakest-link/best-shot winning rule (Theorem \ref{thm:WL}). That is, one of the players must send more resources of \emph{every} type in order to win a contest (weakest-link), while the other player only needs to send more of any one resource type (best-shot).
    
    \item In our second main contribution, we consider a scenario where utilizing resources incurs a cost. We derive a unique equilibrium that describes the optimal investments into each type of resource for both players (Theorem \ref{thm:investment}).

	\item Our third contribution provides a complete equilibrium characterization in the case that the winning rule is given by a weighted linear combination of the allocated resource types (Proposition \ref{thm:WC}). That is, each resource type has a relative effectiveness against all other types.
\end{itemize}
Numerical experiments are also included throughout to illustrate the main results. Figure \ref{fig:setup} provides a summary of these contributions.

\subsection{Related works}

Single-resource allocation models constitute the vast majority of studies in the Colonel Blotto literature.  This body of work provides unique insights into  many  aspects of adversarial interactions,  such as multi-agent coalition formation \cite{Kovenock_2012,Gupta_2014a,Gupta_2014b}, incomplete and asymmetric information \cite{Ewerhart_2021,Paarporn_2022_LCSS,paarporn2024incomplete}, and settings with networked contests and players \cite{Kovenock_2018,Guan_2019,diaz2023beyond}. These are often applied to specific applications, such as influence of social networks  \cite{masucci2014strategic}, market competition \cite{maljkovic2024blotto}, political races \cite{Behnezhad_2018}, and cybersecurity \cite{chia2011colonel}.

There are few studies that have explicitly considered  multiple and heterogeneous resource types. A multi-resource integer Blotto game was considered in \cite{Behnezhad_2017}, which focused on deriving efficient computational algorithms rather than studying specific winning rules and the corresponding equilibrium payoffs. Recent interest in defense applications proposes a framework for ``mosaic warfare", where multiple resource types have heterogeneous capabilities and effectiveness against the opponent's resource types \cite{grana2021findings}. The effect of fractionated resource types (quantized amounts of resources) has also been considered when players have access to different types of resources \cite{lamb2022benefits}. An analysis of tradeoffs between two types of resources that are allocated in different time periods, i.e. pre-allocated and real-time resource types, is recently given in \cite{Vu_EC2021,paarporn2024reinforcement} for General Lotto games.

\smallskip \noindent \textbf{Notation:}  We denote the set of non-negative vectors of length $n$ as $\R_{\geq 0}^n$. We will use bold lettering to denote vector variables. The function $\ones\{E\}$ is the indicator function on an event $E$: it is 1 if $E$ is true, and 0 otherwise. We denote $\Delta(S)$ as the collection of all probability distributions over the elements of an arbitrary set $S$.

\section{Problem formulation}\label{sec:model}

In this section, we first review the classic General Lotto game, which considers the allocation of a single resource type. We then formulate a novel setting where the players' depends on the allocation of multiple resource types.

\subsection{Classic, Single-Resource General Lotto Game}

Two players $\mcx$ and $\mcy$ compete over a collection of simultaneous valuable contests, $\mcc = \{1,2,\ldots,C\}$. The contests have values $\bs{v} = (v_1,\ldots,v_C) \in \Rp^C$. Without loss of generality, the values are normalized such that $\sum_{c\in\mcc} v_c = 1$. The players have access to a common single-dimensional resource, e.g. only money, security forces, or human resources, that they use to compete. A resource allocation for player $\mcx$ is a vector  $\btx = (x_1,\ldots,x_C) \in \Rp$, and similarly $\bty = (y_1,\ldots,y_C) \in \Rp$ for player $\mcy$. The amount $x_c$ is interpreted as the quantity of resources allocated to contest $c \in \mcc$. A player wins a contest by allocating more resources than the opponent. Thus, the payoff to player $\mcx$ is defined as
\be
	\pi_\mcx(\btx, \bty) \triangleq  \sum_{c\in \mcc} v_c \cdot \ones\{x_c \geq y_c \},
\ee
and the payoff to player $\mcy$ is defined as
\be
	\pi_\mcy(\btx, \bty) \triangleq  \sum_{c\in \mcc} v_c \cdot \ones\{x_c < y_c \} = 1 - \pi_\mcx(\btx, \bty).
\ee
In this formulation, each player is able to randomize its allocation, but cannot allocate an amount of resources that exceeds a fixed budget \emph{in expectation}. Player $\mcx$ has a fixed budget $X \geq 0$ and player $\mcy$ has a fixed budget $Y \geq 0$. An admissible strategy for player $\mcx$ is a cumulative distribution function $F_\mcx$ over $\Rp^C$ whose expected value does not exceed $X$. That is, the strategy space for player $\mcx$ is
\be\label{eq:GL_admissible}
	\F(X) \triangleq  \left\{ F_\mcx \in \Delta(\Rp^C) : \E_{\btx\sim F_\mcx}\left[\sum_{c\in\mcc} x_c\right] \leq X \right\}
\ee
and the strategy space for player $\mcy$ is $\F(Y)$. With slight abuse of notation, we denote the players' expected payoffs with respect to a strategy profile $(F_\mcx,F_\mcy)$ as
\be\label{eq:EU_GL}
    \pi_\mcx(F_\mcx,F_\mcy) = \E_{\substack{\btx \sim F_\mcx \\ \bty \sim F_\mcy}}\left[ \pi_\mcx(\btx,\bty) \right]. 
\ee
and $\pi_\mcy(F_\mcx,F_\mcy) = 1 - \pi_\mcx(F_\mcx,F_\mcy)$ is similarly defined. These payoff functions define a two-player simultaneous-move game, termed the \emph{General Lotto game}, which we will refer to as $\GL(X,Y;\bv)$. A diagram of this setup is provided in Figure \ref{fig:setup} (Left).

\begin{definition}\label{def:equil}
	An \emph{equilibrium} is a strategy profile $(F_\mcx^*,F_\mcy^*)$ that satisfies
	\be
		\pi_\mcx(F_\mcx,F_\mcy^*) \leq \pi_\mcx(F_\mcx^*,F_\mcy^*) \leq \pi_\mcx(F_\mcx^*,F_\mcy)
	\ee
	for any $F_\mcx \in \F(X)$ and for any $F_\mcy \in \F(Y)$.
\end{definition}
In an equilibrium, player $\mcx$ cannot increase its payoff by unilaterally changing its strategy, and player $\mcy$ cannot decrease player $\mcx$'s payoff by unilaterally changing its strategy. Provided below is a well-known result from the literature detailing the unique payoffs that the players obtain in any equilibrium of the General Lotto game.

\begin{theorem}[Adapted from~\cite{Hart_2008}]\label{thm:GL}
    Consider a General Lotto game $\GL(X,Y;\bv)$. The unique equilibrium payoff for player $\mcx$ is
    \be\label{eq:GL_equil} 
    	\pi_\mcx^*(X,Y) \triangleq L\left( \frac{Y}{X} \right),
    \ee
    where $L: \Rp \rightarrow (0,1]$ is defined as
    \be\label{eq:lotto}
    	L(\alpha) \triangleq
	\begin{cases}
		1 - \frac{\alpha}{2}, &\text{if } \alpha \leq 1 \\
		\frac{1}{2\alpha}, &\text{if } \alpha > 1 \\
	\end{cases}.
    \ee
    The equilibrium payoff to player $\mcy$ is $\pi_\mcy^*(X,Y) \triangleq 1 - \pi_\mcx^*(X,Y)$. 
    
\end{theorem}

It is important to note that in $\GL(X,Y;\bv)$, the equilibrium payoff depends only on the ratio of the players' budgets, and the total value of all contests $\sum_c v_c$. Here, since $\bv$ is normalized, the total value is just 1, but in cases where $\bv$ is not normalized, the equilibrium payoff would simply be $(\sum_c v_c) L(Y/X)$. A plot of the equilibrium payoff $\pi_\mcx^*$ is shown in Figure \ref{fig:contour} (Left).

\subsection{Multi-resource General Lotto game}

We now introduce an extension of the classic General Lotto game to a scenario where the players have multiple types of resources at their disposal (e.g. money, human resources, ...), and their success on any contest depends on the combination of allocations of all resource types. Suppose there are $T \geq 1$ distinct resource types, which we enumerate as $\mct = \{1,\ldots, T\}$. Let $\bX \triangleq (X_1,\ldots,X_T)$ and $\bY \triangleq (Y_1,\ldots,Y_T)$ be player $\mcx$ and $\mcy$'s resource budgets for each type, respectively. Here, an allocation for player $\mcx$, $\mcy$ is 
\be
	\ba
		\btx &= (\bx_1,\ldots,\bx_C) \in \Rp^{CT} \\
		\bty &= (\by_1,\ldots,\by_C) \in \Rp^{CT} \\
	\ea
\ee
where for each contest $c\in\mcc$, $\bx_c = (x_{c,1}, \ldots, x_{c,T}) \in \Rp^T$. We note that $T=1$ recovers the classic setup.

\begin{figure*}[t]
    \centering
    \includegraphics[scale=0.35]{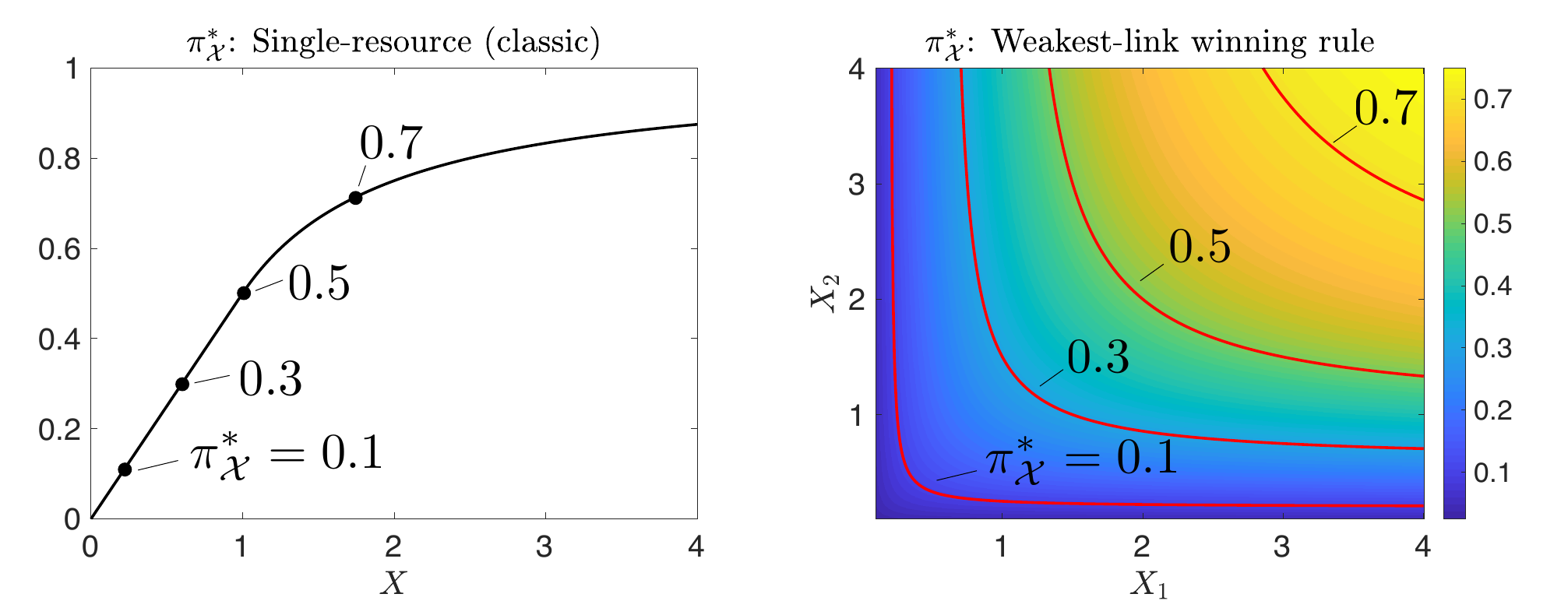}
    \includegraphics[scale=0.3]{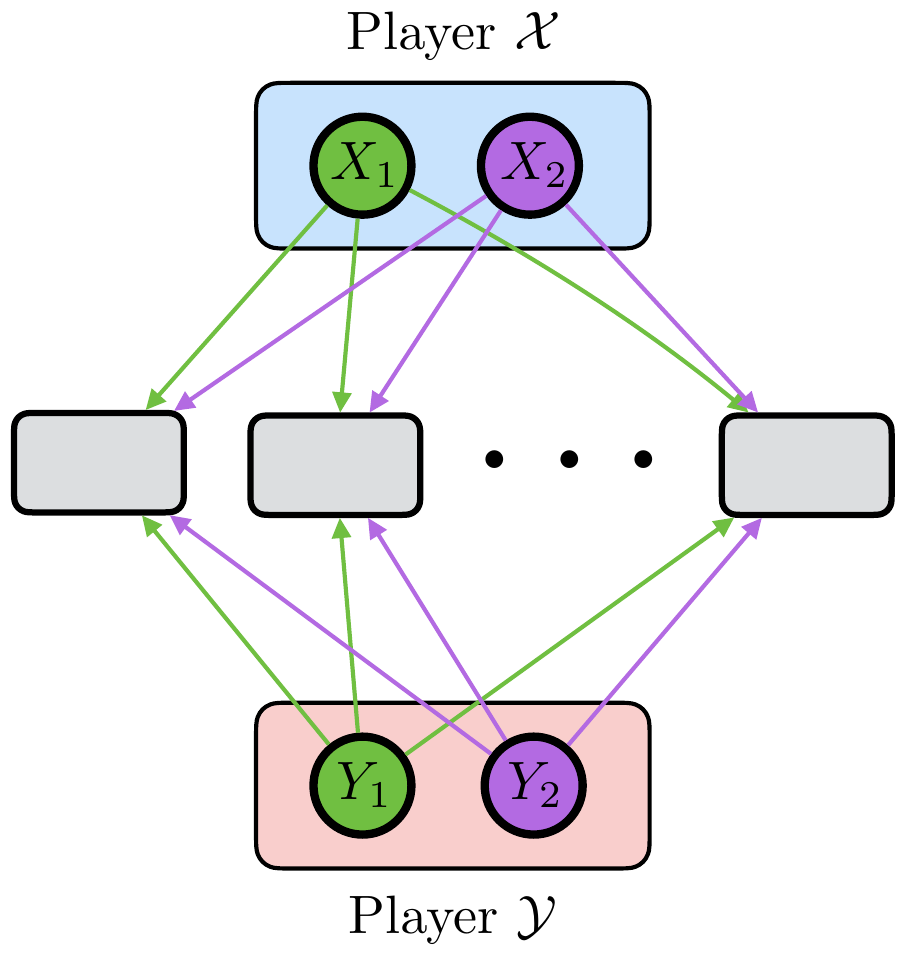}
    \caption{ (Left) This plot shows the equilibrium payoff to player $\mcx$ in the classic single-resource General Lotto game (Theorem \ref{thm:GL}), as a function of its resource budget $X \geq 0$. In this plot, we fix player $\mcy$'s budget $Y=1$. Note that for any given performance level, e.g. a payoff of 0.1, there is a \emph{unique} budget for $\mcx$ that achieves this performance level. (Center) This plot shows the equilibrium payoff to player $\mcx$ in the two-resource General Lotto game with the weakest-link winning rule (Theorem \ref{thm:WL}). For a fixed performance level, e.g. 0.1, there is now a \emph{contour} of budget pairs $(X_1,X_2)$ that achieves this payoff. Here, we fix budgets $Y_1 = Y_2 = 1$. (Right) Depiction of the two-resource game from the center figure.}
    \label{fig:contour}
\end{figure*}

The payoff to player $\mcx$ is
\be
	\pi_\mcx(\btx,\bty) \triangleq \sum_{c\in\mcc} v_c\cdot W(\bs{x}_c,\bs{y}_c)
\ee
where $W : \Rp^T \times \Rp^T \rightarrow \{0,1\}$ is the \emph{winning rule} for player $\mcx$. The payoff to player $\mcy$ is then 
\be
	\pi_\mcy(\btx,\bty) \triangleq \sum_{c\in\mcc} v_c\cdot (1-W(\bs{x}_c,\bs{y}_c)) = 1 - \pi_\mcx(\btx,\bty).
\ee
The winning rule $W$ determines the winner of contest $c\in\mcc$ based on the composition of the allocations of all resource types from both players, $\bx_c, \by_c$. This significantly differs from the classic setup that utilizes the ``winner-take-all" winning rule: whoever allocated more resources (of homogeneous type) wins the contest. An illustration that contrasts these setups is depicted in Figure \ref{fig:setup}. We will focus on the following winning rules defined below.

\begin{itemize}
    \item \textbf{Weakest-link:} Player $\mcx$ is required to allocate more than $\mcy$ for \emph{all resource types},
    \begin{equation}\label{eq:WL}
        W_{\WL}(\bs{x},\bs{y})  \triangleq \ones\left\{ x_t \geq y_t, \ \forall t \in \mct \right\}.
    \end{equation}
    
    \item \textbf{Best-shot:}  Player $\mcx$ is required to allocate more than $\mcy$ for \emph{at least one resource type},
    \be\label{eq:BS}
        W_{\BS}(\bs{x},\bs{y})  \triangleq \ones\left\{ x_t \geq y_t,  \text{ for some } t \in \mct \right\}.
    \ee
    
\end{itemize}
Note that a weakest-link winning rule for $\mcx$, $\pi_\mcx(\bs{x},\bs{y}) = \sum_{c} v_c W_{\WL}(\bs{x},\bs{y})$, necessarily implies $\mcy$ has a best-shot winning rule, $\pi_\mcy(\bs{x},\bs{y}) = \sum_{c} v_c W_{\BS}(\bs{y},\bs{x})$, and vice versa. The weakest-link/best-shot winning rules reflect interactions between a defender and attacker. For example, a defender is required to have superiority on all fronts in order to ensure security of a networked system or critical infrastructure -- e.g. ensuring security against malware, denial-of-service, and social engineering attacks.  


To complete the setup, we consider an admissible strategy for player $\mcx$ to be any distribution $F_\mcx$ over $\Rp^{CT}$ such that the expected total allocation of resource type $t$ across all contests does not exceed $X_t$. That is, the strategy space for player $\mcx$ is
\be\label{eq:ML_admissible}
	\F(\bX) \triangleq \left\{ F_\mcx \in \Delta(\Rp^{CT}) : \E_{\btx \sim F_\mcx}\left[\sum_{c\in\mcc} x_{c,t} \right] \leq X_t, \ \forall t \in \mct \right\}
\ee
and the strategy space for player $\mcy$ is $\F(\bY)$. In the same way that expected payoffs were defined in \eqref{eq:EU_GL}, the expected payoffs with respect to a strategy profile $(F_\mcx,F_\mcy) \in \F(\bX)\times\F(\bY)$ is given by
\be\label{eq:EU_ML}
	\pi_\mcx(F_\mcx,F_\mcy) = \E_{\substack{\btx\sim F_\mcx \\ \bty \sim F_\mcy}}\left[ \sum_{c\in\mcc} v_c\cdot W(\bs{x}_c,\bs{y}_c) \right],
\ee
with $\pi_\mcy(F_\mcx,F_\mcy)  = 1 - \pi_\mcx(F_\mcx,F_\mcy)$. This defines a two-player simultaneous-move game, which we will refer to as a \emph{Multi-resource General Lotto game}, and denote it as $\ML(\bs{X},\bs{Y};W,\bv)$. An \emph{equilibrium} $(F_\mcx^*,F_\mcy^*)$ in this game is analogously defined as in Definition \ref{def:equil}.
 
\section{Results: Equilibrium characterization of ML game}\label{sec:results}

In this section, we present our first main contribution of this paper, which provides the full equilibrium characterization for the multi-resource General Lotto game under the weakest-link/best-shot winning rule. The results detail the unique equilibrium payoffs to both players as well as the equilibrium strategy profiles.

\begin{theorem}\label{thm:WL}
    Consider a Multi-resource General Lotto game under the weakest-link winning rule, $\ML(\bX,\bY;W_\WL,\bv)$. The unique equilibrium payoff to player $\mcx$ is    
    \be\label{eq:WL_equil}
        \pi_\WL^*(\bX,\bY) \triangleq L(\alpha(\bX,\bY) )
    \ee
    where $\alpha(\bX,\bY) \triangleq \sum_{t\in\mct} \frac{Y_t}{X_t}$. The equilibrium payoff to $\mcy$ is $\pi_\BS^*(\bX,\bY) \triangleq 1 - \pi_\WL^*(\bX,\bY)$. An equilibrium strategy profile $(F_\mcx^*,F_\mcy^*) \in \F(\bX) \times \F(\bY)$ is given as follows. If $\alpha(\bX,\bY) \leq 1$, then for all $c \in \mcc$ and any $\bs{u} \in \Rp^T$,
    \be\label{eq:Fstar_X_stronger}
        \ba
            F_{\mcx,c}^*(\bs{u}) &= \min\left\{ \min\left\{ \frac{u_t}{2 v_c X_t }, 1 \right\} \right\}_{t \in \mct} \\
            F_{\mcy,c}^*(\bs{u}) &= 1-\alpha(\bX,\bY) \\
            &\quad + \alpha(\bX,\bY) \sum_{t \in \mct} \frac{Y_t/X_t}{\alpha(\bX,\bY)} \cdot \min\left\{ \frac{u_t}{2 v_c X_t }, 1 \right\}. \\
        \ea
    \ee
    If $\alpha(\bX,\bY) > 1$, then for all $c \in \mcc$ and any $\bs{u} \in \Rp^T$,
    \be\label{eq:Fstar_X_weaker}
        \ba
            F_{\mcx,c}^*(\bs{u}) &= 1 - \frac{1}{\alpha(\bX,\bY)} \\
            & +  \frac{1}{\alpha(\bX,\bY)}\cdot \min\left\{ \min\left\{ \frac{u_t}{2 v_c X_t \alpha(\bX,\bY) }, 1 \right\} \right\}_{t \in \mct} \\
            F_{\mcy,c}^*(\bs{u}) &= \sum_{t \in \mct} \frac{Y_t/X_t}{\alpha(\bX,\bY)} \cdot \min\left\{ \frac{ u_t}{2 v_c X_t \alpha(\bX,\bY)}, 1 \right\}. \\
        \ea
    \ee
\end{theorem}

Section \ref{sec:proofs} details the full proof of Theorem \ref{thm:WL} by developing a series of intermediate Lemmas.

We first give some remarks about the form of the equilibrium payoff \eqref{eq:WL_equil}. We note that it takes the same form $L(\alpha)$ of the equilibrium payoff from the classic General Lotto game (Theorem \ref{thm:GL}), but instead the budget ratio $\alpha(\bX,\bY)$ is evaluated to be the summation of the budget ratios over all resource types. Thus, we observe that under the weakest-link winning rule for $\mcx$, it is necessary that it has a positive resource budget $X_t>0$ $\forall t\in\mct$ in order to attain a non-zero payoff, since player $\mcx$ needs to win on all resource types for any given contest. We note this is not the case for player $\mcy$, since it has the best-shot winning rule on any given contest. 

An illustration of the equilibrium payoff is shown in Figure \ref{fig:contour} (Center). A notable feature here is that there is a \emph{contour} of resource budgets $(X_1,X_2)$ satisfying $\frac{Y_1}{X_1} + \frac{Y_2}{X_2} = \alpha$ that achieve the same, fixed equilibrium payoff $L(\alpha)$. The multi-dimensional dependence of the players' performance metrics suggests one can evaluate the most cost-effective investments in combinations of resources to achieve a given performance level. We will address this problem in Section \ref{sec:investment}.


We also give a brief discussion of the equilibrium strategies here, but we direct the reader to Section \ref{sec:proofs} for their full interpretations and derivations. With strategy $F_{\mcx}^*$ \eqref{eq:Fstar_X_stronger}, player $\mcx$ chooses its allocation to contest $c$ according to $x_{c,t} = 2v_c X_t \cdot U$ for each $t \in \mct$, where $U \in [0,1]$ is an independent uniform random sample. Given $\mcx$ has the weakest-link rule, this strategy ensures that a positive amount of every resource type is allocated to each contest. With strategy $F_\mcy^*$, player $\mcy$ chooses its allocation $\bs{y}_c$ to contest $c$ as follows. With probability $\alpha(\bX,\bY)$, it randomly selects a resource type $t \in \mct$ with probability $\frac{Y_t/X_t}{\alpha(\bX,\bY)}$. Then, it allocates an amount $y_{c,t} = 2 v_c X_t \cdot U$ of type $t$ resources, where $U \in [0,1]$ is an independent uniform random sample, and does not allocate any other resource types, i.e. $y_{c,t'} = 0$ for all $t' \neq t$. While the best-shot rule specifies that $\mcy$ can win by allocating more on \emph{at least one} resource  type, the equilbrium strategy $F_\mcy^*$ attempts to win the contest by winning on only one resource type.




\section{Strategic investment of multiple resources}\label{sec:investment}

In this section, we examine how valuable each resource type is by considering a scenario where the players must decide how much of each type to invest in, given there are costs of investment. We formulate a two-stage interaction where the players make investment decisions in the first stage to acquire resources, and then engage in a multi-resource Lotto game in the second stage. For resources of type $t \in \mct$, player $\mcx$ pays a per-unit cost $\costx_t > 0$ and player $\mcy$ pays a per-unit cost $\costy_t$. The interaction unfolds as follows.

\noindent\textbf{Stage 1:} Both players simultaneously decide their resource investments, $\bX$ and $\bY$. 

\noindent\textbf{Stage 2:} The players engage in a multi-resource General Lotto game $\ML(\bX,\bY;W,\bv)$.


In the case that the winning rule is weakest-link, the final payoffs obtained are
\be
	U_\mcx(\bX,\bY) \triangleq \pi^*_\WL(\bX,\bY) - \sum_{t\in\mct} \costx_t \cdot X_t
\ee
for player $\mcx$, and
\be
	U_\mcy(\bX,\bY) \triangleq \pi^*_\BS(\bX,\bY) - \sum_{t\in\mct} \costy_t \cdot Y_t
\ee
for player $\mcy$. Here, each player uses its resources invested from Stage 1 to compete in the simultaneous-move game $\ML$ in stage 2, and the resulting payoffs from stage 2 are the unique equilibrium payoffs of $\ML(\bX,\bY;W,\bv)$ characterized in Theorems \ref{thm:WL} and \ref{thm:WC}.  Consequently, the only strategic decisions that need to be investigated are how the players should invest in Stage 1. This two-stage interaction can thus be viewed as a two-player strategic-form game with strategy spaces $\bX \in \Rp^T$ for $\mcx$ and $\bY\in\Rp^T$ for $\mcy$. We denote this game as the \emph{Multi-resource General Lotto game with costs}, $\MLC(\bs{\costx},\bs{\costy},W)$.

Our goal is to identify the equilibria of $\MLC(\bs{\costx},\bs{\costy},W)$. Here, an equilibrium is an investment profile $(\bX^*,\bY^*)\in\Rp^T \times \Rp^T$ that satisfies
\be
	\ba
		U_\mcx(\bX^*,\bY^*) &\geq U_\mcx(\bX,\bY^*), \ \forall \bX \in \Rp^T \\
		U_\mcy(\bX^*,\bY^*) &\geq U_\mcx(\bX^*,\bY), \ \forall \bY \in \Rp^T \\
	\ea
\ee

Note that unlike the General Lotto games GL and ML, the strategic interaction with costs is not a constant-sum game because  the associated costs of investment are included in the players' payoff functions. The result below provides the equilibrium investments and payoffs for both players under the weakest-link winning rule.
\begin{theorem}\label{thm:investment}
	Consider the Multi-resource General Lotto game with costs under the weakest-link winning rule $\MLC(\bs{\costx},\bs{\costy},W_\WL)$. It admits a unique equilibrium investment profile $(\bX^*,\bY^*)$, wherein both players spend an identical amount of money, i.e.
    \be
        \sum_{t\in\mct} \costx_t X_t^* = \sum_{t\in\mct} \costy_t Y_t^*.
    \ee
    The unique equilibrium investment profile is given as follows. If $r \triangleq \sum_{t\in\mct} \frac{\costx_t}{\costy_t} > 1$, then
    \be
        X_t^* = \frac{1}{2 \costy_t r^2}, \quad Y_t^* = \frac{\costx_t}{2 \costy_t^2 r^2} \quad \forall t\in\mct.
    \ee
    yielding equilibrium payoffs of $U_\mcx^* \triangleq U_\mcx(\bX^*,\bY^*) = 0$ and $U_\mcy^* \triangleq U_\mcy(\bX^*,\bY^*) = 1 - \frac{1}{2r}$. 
    
    \noindent If $r \triangleq \sum_{t\in\mct} \frac{\costx_t}{\costy_t} \leq 1$, then
    \be
        X_t^* = \frac{1}{2\costy_t}, \quad Y_t^* = \frac{\costx_t}{2\costy_t^2} \quad \forall t\in\mct.
    \ee
    yielding equilibrium payoffs of $U_\mcx^*  = 1 - \frac{r}{2}$ and $U_\mcy^*  = 0$.
\end{theorem}

Interestingly, regardless of the cost parameters, both players spend an equivalent amount of money in  equilibrium. However, the total amount of resources purchased is not necessarily identical. Player $\mcx$ purchases $X_{tot}^* = \sum_t X_t^* = \frac{1}{2r^2}\sum_t \frac{1}{\costy_t}$ total resources, and player $\mcy$ purchases $Y_{tot}^* =  \frac{1}{2r^2}\sum_t \frac{\costx_t}{\costy_t^2}$ total resources.  We note that both of these quantities depend on the cost parameters $\bs{\kappa},\bs{\sigma}$.

The players' equilibrium payoffs have a discontinuity when the parameter $r=1$. This is unlike the General Lotto games GL and ML, where although there are two separate cases, the equilibrium payoffs are still continuous in the $\alpha$ parameter. Figure \ref{fig:investments} (Top Right) provides a plot of the players' equilibrium payoffs.

The fraction of each resource type relative to the total resources purchased is also different among the players. For player $\mcx$, the fraction invested in type $t$ is $\frac{X_t^*}{X_{tot}^*}=\frac{1/\costy_t}{\sum_{t'} 1/\costy_{t'}}$. Interestingly, these fractions only depend on the cost parameters of $\mcy$, and do not depend on its own cost parameters $\bs{\kappa}$. For player $\mcy$, the fraction invested in type $t$ is $\frac{Y_t^*}{Y_{tot}^*}=\frac{\costx_t/\costy_t^2}{\sum_{t'} \costx_{t'}/\costy_{t'}^2}$. These fractions depend on both cost parameters $\bs{\kappa},\bs{\sigma}$. A numerical case study that illustrates the equilibrium resource investments is given in the plots of Figure \ref{fig:investments}.


\begin{figure}[t]
    \centering
    \includegraphics[width=0.9\linewidth]{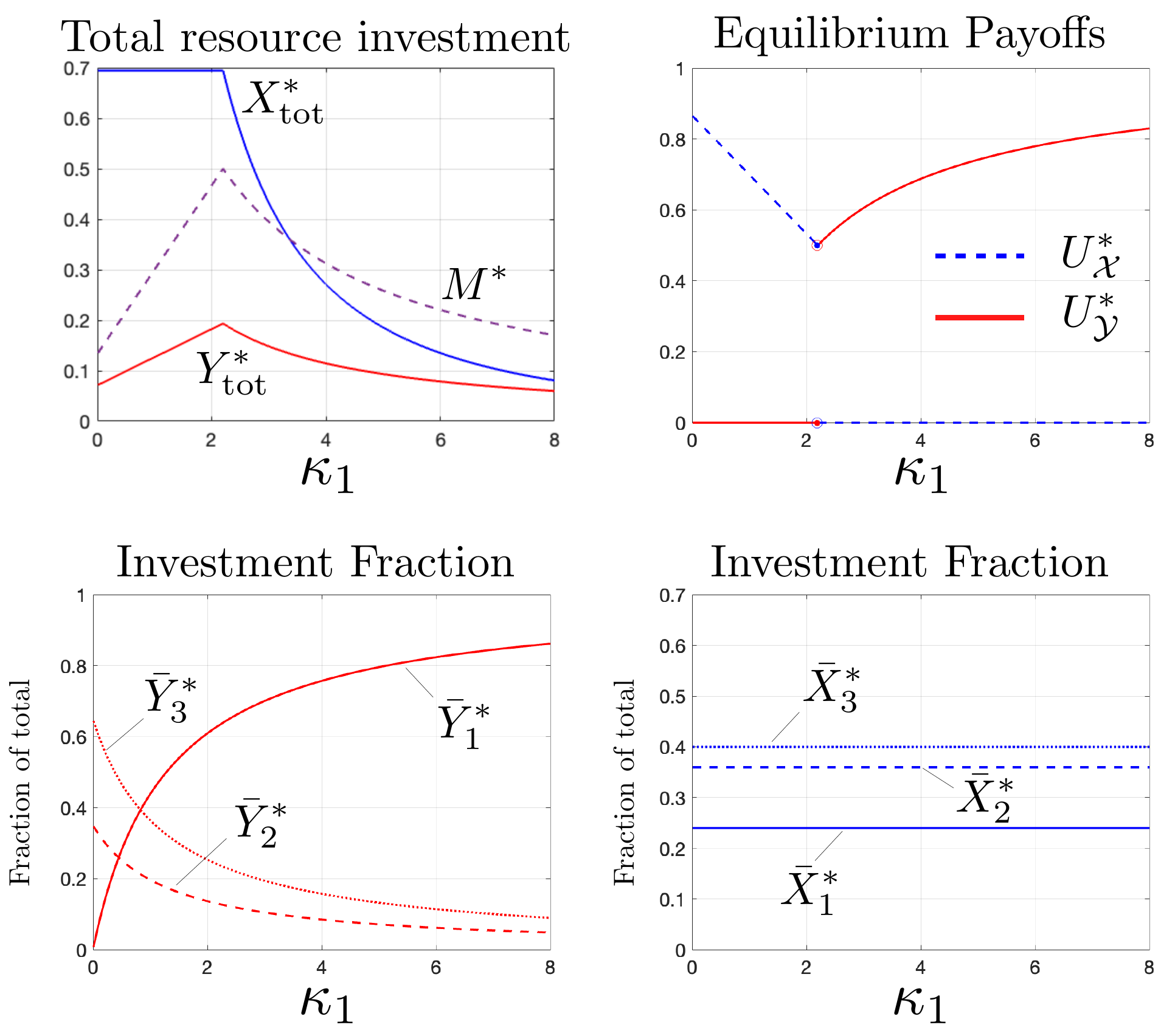}
    \caption{These plots illustrate the equilibrium properties of the ML-C game established in Theorem \ref{thm:investment} through a simulation example with three resource types. We fix costs $\bs{\sigma} = (3, 2, 1.8)$ for player $\mcy$, and vary the cost of only resource 1 for player $\mcx$: $\bs{\sigma} = (\kappa_1, 0.2, 0.3)$, with $\kappa_1 \geq 0$ as the x-axis for all plots above.  (Top Left) The players' total resource investment in equilibrium, denoting $X_{\text{tot}}^* = \sum_{t=1}^3 X_t^*$ and $Y_{\text{tot}}^* = \sum_{t=1}^3 Y_t^*$. Interestingly, player $\mcx$ does not change its total resource investment $X_{\text{tot}}^*$ for low costs, $\kappa < 2.2$. The dashed line $M^*=\sum_{t=1}^3 \kappa_t X_t^* = \sum_{t=1}^3 \sigma_t Y_t^*$ denotes the money spent by each player, which is the same for both. (Top Right) Equilibrium payoffs to both players. Increasing $\kappa_1$ linearly increases the ratio $r=\sum_{t\in\mct} \frac{\costx_t}{\costy_t}$, which determines the equilibrium payoff. In this example $r = 1$ when $\kappa_1 = 2.2$. (Bottom Left) Player $\mcy$'s investment fractions in the three resource types. As resource 1 becomes more expensive for player $\mcx$, player $\mcy$ takes advantage by investing more into resource 1. Here, we denote $\bar{Y}_i^* = Y_i^*/Y^*_{tot}$. (Bottom Right) Player $\mcy$'s investment fractions in the three resource types. These fractions remain constant since they only depend on the cost parameters $\bs{\sigma}$ of player $\mcy$. }
    \label{fig:investments}
\end{figure}


To establish the proof of Theorem \ref{thm:investment}, we first approach the problem of finding the equilibrium resource investments when the players are given fixed use-it-or-lose-it monetary budgets $M_\mcx, M_\mcy \geq 0$ to spend on resources. In this scenario, the monetary budgets are viewed as sunk costs, so that  player $\mcx$ can choose any $\bX$ for which $\sum_{t\in\mct} \costx_t X_t = M_\mcx$ with the objective of maximizing only $\pi_\WL^*(\bX,\bY)$. Similarly, player $\mcy$ can choose any $\bY$ for which $\sum_{t\in\mct} \costy_t Y_t = M_\mcy$ with the objective of maximizing only $\pi_\BS^*(\bX,\bY)$. We denote this constrained game as $\MLM(M_\mcx, M_\mcy,\bs{\costx},\bs{\costy},W_\WL)$.

\begin{lemma}\label{lem:sunk}
	With sunk costs $M_\mcx, M_\mcy$, the equilibrium investments $(\bX^*,\bY^*)$ in the game $\MLM(M_\mcx, M_\mcy,\bs{\costx},\bs{\costy},W_\WL)$ are given by
	\be\label{eq:sunk_division}
		X_t^* = \frac{M_\mcx}{\costx_t} \frac{\costx_t/\costy_t}{r},  \quad Y_t^* = \frac{M_\mcy}{\costy_t} \frac{\costx_t/\costy_t}{r}, \ \forall t \in \mct.
	\ee
	The equilibrium payoffs are
	\be\label{eq:sunk_payoff}
		\ba
			\pi_\WL^*(\bX^*,\bY^*) &= L\left( \frac{M_\mcy}{M_\mcx} r \right) \\
			\pi_\BS^*(\bX^*,\bY^*) &= 1 - L\left( \frac{M_\mcy}{M_\mcx} r \right). \\
		\ea
	\ee
\end{lemma}

The proof of this intermediate result is deferred to the Appendix. Lemma \ref{lem:sunk} asserts that if any positive amount of monetary investment is spent, then in an equilibrium, the proportions of that investment that go towards each resource type is given precisely by \eqref{eq:sunk_division}. The resulting payoffs are then given in \eqref{eq:sunk_payoff}. This result allows us to view the strategic decision-making in Stage 1 of $\MLC(\bs{\costx},\bs{\costy},W_\WL)$ simply as a monetary investment $M_\mcx, M_\mcy \geq 0$.

Thus, the problem of finding equilibria of $\MLC(\bs{\costx},\bs{\costy},W_\WL)$ can be re-cast as finding the equilibria of the game associated with payoff functions
\be
	\ba
		U_\mcx(M_\mcx,M_\mcy) &= L\left( \frac{M_\mcy}{M_\mcx} r \right) - M_\mcx \\
		U_\mcy(M_\mcx,M_\mcy) &= 1- L\left( \frac{M_\mcy}{M_\mcx} r \right) - M_\mcy \\
	\ea
\ee
and strategies $M_\mcx \geq 0$, $M_\mcy \geq 0$. We are now able to proceed with the proof of Theorem \ref{thm:investment}.

\begin{proof}[Proof of Theorem \ref{thm:investment}]
	Our approach is to derive the best-response curves for both players. For a fixed $M_\mcx \geq 0$,
	\be
		\text{BR}_\mcy(M_\mcx) := \max_{M_\mcy} U_\mcy(M_\mcx, M_\mcy)
	\ee
	Omitting the algebraic steps, we can precisely characterize this as
	\be
		\text{BR}_\mcy(M_\mcx) =
		\begin{cases}
			\sqrt{\frac{M_\mcx}{2r}}, &\text{if } M_\mcx < r/2 \\
			[0,\frac{1}{2}], &\text{if } M_\mcx = r/2 \\
			0, &\text{if } M_\mcx > r/2 \\
		\end{cases}
	\ee
	In a similar manner, we derive the best-response for player $\mcx$ as
	\be
		\text{BR}_\mcx(M_\mcy) =
		\begin{cases}
			\sqrt{\frac{M_\mcy r}{2}}, &\text{if } M_\mcx < \frac{1}{2r} \\
			[0,\frac{1}{2}], &\text{if } M_\mcy = \frac{1}{2r} \\
			0, &\text{if } M_\mcy > \frac{1}{2r} \\
		\end{cases}
	\ee
	An equilibrium is any point $(M_\mcx^*,M_\mcy^*)$ at which $M_\mcy^* = \text{BR}_\mcy(M_\mcx^*)$ and  $M_\mcx^* = \text{BR}_\mcx(M_\mcy^*)$, i.e. where the two graphs of $\text{BR}_\mcy$ and $\text{BR}_\mcx$  intersect. If $r > 1$, then $\text{BR}_\mcx(M_\mcy)$ intersects $\text{BR}_\mcy(M_\mcx)$ uniquely at the point $(\frac{1}{2r},\frac{1}{2r})$. If $r \geq 1$, then they intersect uniquely at the point $(\frac{r}{2},\frac{r}{2})$. 
	
	The full equilibrium investment profile $(\bX^*,\bY^*)$ can then be deduced using the result from Lemma \ref{lem:sunk}.

\end{proof}

\section{The Weighted Contribution Winning Rule}\label{sec:WC}

In this section, we consider an alternate winning rule that we refer to as a \emph{weighted contribution} rule. Here, each resource type is associated with an effectiveness weight $\bs{a} = (a_1,\ldots,a_T) \in \Rp^T$ for player $\mcx$, and $\bs{b} = (b_1,\ldots,b_T) \in \Rp^T$ for player $\mcy$. The winning rule for player $\mcx$ is
\be\label{eq:WC}
    W_{\WC}(\bs{x},\bs{y})  \triangleq \ones\left\{ \sum_{t\in\mct} a_t x_t \geq  \sum_{t\in\mct} b_t y_t \right\}.
\ee
The weighted contribution winning rule $W_{\WC}$ reflects the heterogeneity of resources with respect to their relative effectiveness for winning a contest. 

We state the full equilibrium characterization of the ML game under the weighted contribution rule in the Proposition below.

\begin{proposition}\label{thm:WC}
    Consider a Multi-resource General Lotto game under the weighted contribution winning rule, $\ML(\bX,\bY;W_\WC,\bv)$, with parameters $\{\bs{a},\bs{b}\}$. The unique equilibrium payoff to player $\mcx$ is 
    \begin{equation}
        \pi_\mcx^* =  L(\beta(\bX,\bY) ),
    \end{equation}
    where $\beta(\bX,\bY) \triangleq \frac{\sum_{t=1}^T b_tY_t}{\sum_{t=1}^T a_tX_t}$. The equilibrium payoff to player $\mcy$ is $\pi_\mcy^* = 1 - \pi_\mcx^*$. An equilibrium strategy profile $(F_\mcx^*,F_\mcy^*) \in \F(\bX) \times \F(\bY)$ is given as follows. If $\beta(\bX,\bY) \leq 1$, then
    \be
        \ba
            F_{\mcx,c}^*(\bs{u}) &= \min\left\{ \min\left\{ \frac{u_t}{2 v_c X_t }, 1 \right\} \right\}_{t \in \mct} \\
            F_{\mcy,c}^*(\bs{u}) &= 1-\frac{1}{\beta} + \frac{1}{\beta}\min\left\{ \min\left\{ \frac{u_t}{2 v_c Y_t \beta }, 1 \right\} \right\}_{t \in \mct} \\
        \ea
    \ee
    If $\beta(\bX,\bY) > 1$, then
    \be
        \ba
            F_{\mcx,c}^*(\bs{u}) &= 1-\frac{1}{\beta} +\frac{1}{\beta} \min\left\{ \min\left\{ \frac{u_t}{2 v_c X_t \beta }, 1 \right\} \right\}_{t \in \mct} \\
            F_{\mcy,c}^*(\bs{u}) &= \min\left\{ \min\left\{ \frac{u_t}{2 v_c Y_t  }, 1 \right\} \right\}_{t \in \mct} \\
        \ea.
    \ee
    
    
\end{proposition}
Unlike the weakest-link formulation, player $\mcx$ can attain a non-zero payoff by only possessing a single resource type. The above result is comparable to the classic, single-resource General Lotto game where the ``effective total budgets" are given by the linear combinations $X=\sum_{t=1}^T a_t X_t$ and $Y=\sum_{t=1}^T b_t Y_t$. Indeed, observe that the equilibrium payoffs coincide with that of the classic single-resource game $L(Y/X)$ (Theorem \ref{thm:GL}) using the effective total budgets $X,Y$.


We omit a full proof of this result since it largely follows the same techniques that we develop to establish the proof of Theorem \ref{thm:WL} in Section \ref{sec:proofs}. In particular, one can establish the proof using the same sequence of Lemmas stated in Section \ref{sec:proofs}. The main difference is that instead of applying the CDF $\mbb{P}_{\mcy}[ y_{c,t} \leq x_{c,t}, \forall t]$ in the calculations of expected utility (i.e. starting in Lemma \ref{lem:UB}, \eqref{eq:UB_calculation}), we apply the CDF $\mbb{P}_{\mcy}[  \sum_{t\in\mct} b_t y_{c,t} \leq \sum_{t\in\mct} a_t x_{c,t}]$ that corresponds to the winning rule $W_{\text{WC}}$.

\section{Equilibrium Analysis and Derivations}\label{sec:proofs}

In this section, we establish the proofs for the main result (Theorem \ref{thm:WL}) regarding equilibrium characterizations of the Multi-resource game, $\ML(\bX,\bY;W_{\text{WL}},\bv)$ under the weakest-link/best-shot winning rules. Below, we describe a rough outline of the technical approach we take to prove Theorem \ref{thm:WL}.
\begin{itemize}[leftmargin=*]
	\item We propose and define two classes of randomized allocation strategies, $\F_\WL$ for the player with weakest-link rule ($\mcx$), and $\F_\BS$ for the player with the best-shot rule ($\mcy$). They are given in Definitions \ref{def:CAS_X} and \ref{def:CAS_Y}, and are strict subsets of the strategy set $\F$ \eqref{eq:GL_admissible}.
	\item In Lemmas \ref{lem:UB} and \ref{lem:min_UB}, we establish the smallest upper-bound that player $\mcy$ can impose on the payoff of player $\mcx$ using strategies from $\F_\BS$.
	\item In Lemmas \ref{lem:LB} and \ref{lem:max_LB}, we establish the largest lower-bound that player $\mcx$ can ensure on its own payoff using strategies from $\F_\WL$.
	\item We find that both the upper and lower bounds coincide in value. Because $\ML(\bX,\bY;W_\WL,\bv)$ is a two-player constant-sum game, this establishes an equlibrium strategy profile.
\end{itemize}

We begin  by observing that one can write the expected payoff from \eqref{eq:EU_ML} in terms of marginal distributions of the players' strategies.
\begin{lemma}
	Consider a Multi-resource Lotto game $\ML(\bX,\bY;W,\bv)$. Then for any strategy profile $(F_\mcx,F_\mcy) \in \F(\bX)\times\F(\bY)$,
	\be\label{eq:EU_ML_marginal}
		\pi_\mcx(F_\mcx,F_\mcy) = \sum_{c\in\mcc} v_c \cdot \E_{\substack{\bx_c\sim F_{\mcx,c} \\ \by_c \sim F_{\mcy,c}}}\left[ W(\bs{x}_c,\bs{y}_c) \right].
	\ee
	where for any $\bu_c \in \Rp^T$, the $c$-\emph{marginal distribution} $F_{\mcx,c}$ is given by
	\be
		\ba
			F_{\mcx,c}(\bu_c) &\triangleq \P[x_{c,t} \leq u_{c,t}, \ \forall t\in\mct] \\
			&= \lim_{\bu_d \rightarrow \infty^T, \forall d \neq c} F_\mcx((\bu_d)_{d\in\mcc})
		\ea
	\ee
	Similar definitions apply for $F_{\mcy,c}$.
\end{lemma}

Each expectation term in \eqref{eq:EU_ML_marginal} is equivalent to an expectation with respect to the $T$-variate \emph{marginal distributions} for each contest $c\in\mcc$. Therefore, any full  strategy $F_\mcx \in \F(\bX)$ can be specified in terms of its $c$-marginals, $\{F_{\mcx,c}\}_{c\in\mcc}$.

\subsection{Proposed strategies for weakest-link formulation}

Let us consider the multi-resource game corresponding to the weakest-link winning rule, $\ML(\bX,\bY;W_\WL,\bv)$. For player $\mcx$, who has the weakest-link rule, we propose a particular class of strategies defined below.
\begin{definition}\label{def:CAS_X}
	We define the class of strategies $\F_\WL(\bX) \subset \F(\bX)$ as follows. Any $\hat{F}_\mcx \in \F_\WL(\bX)$ has $c$-marginals that can be written in the form
	\be\label{eq:CAS_X}
		\hat{F}_{\mcx,c}(\bu) = 1 - \dx + \dx\cdot \min\left\{ \min\left\{ \frac{\dx}{2 v_c X_t }\cdot u_t, 1 \right\} \right\}_{t \in \mct}.
	\ee
	for all $c\in\mcc$, $\bu \in \Rp^T$. Here, the number $\dx \in [0,1]$ parameterizes this class of strategies.
\end{definition}
The strategy \eqref{eq:CAS_X} can intuitively be described as follows. To draw a random allocation $\bx_c \in \Rp^T$ from $\hat{F}_{\mcx,c}$, 
\ben
    \i With probability $1-\dx$: allocate zero resources, i.e. $x_{c,t} = 0$ for all $t \in \mct$.
    \i With probability $\dx$: Draw a random $U \sim \text{Unif}[0,1]$. Then, allocate $x_{c,t} = \frac{2v_c X_t}{\dx}\cdot U$ for each $t \in \mct$.
\een
The strategy $\hat{F}_{\mcx,c}$ either allocates zero resources of all types to contest $c$, or allocates a positive amount of every resource type. This strategy allocates $\E_{\bx_c\sim F_{\mcx,c}}[x_{c,t}] = v_c X_t$ resources of type $t$ to contest $c$ in expectation, and thus it is budget-feasible.

For player $\mcy$, who has the best-shot rule, we propose a particular class of strategies defined below.
\begin{definition}\label{def:CAS_Y}
	We define the class of strategies $\F_\BS(\bY) \subset \F(\bY)$ as follows. Any $\hat{F}_\mcx \in \F_\BS(\bY)$ has $c$-marginals that can be written in the form
	\be\label{eq:CAS_Y}
	    \hat{F}_{\mcy,c}(\bu) = 1-\dy + \dy \sum_{t \in \mct} p_t \cdot \min\left\{ \frac{\dy p_t}{2 v_c Y_t }\cdot u_t, 1 \right\}.
	\ee
	for all $c\in\mcc$, $\bu \in \Rp^T$. Here, the number $\dy \in [0,1]$ and probability vector $\bs{p} = (p_t)_{t\in\mct}$ parameterize this class of strategies.
\end{definition}
The strategy \eqref{eq:CAS_Y} can intuitively be described as follows. To draw a random allocation $\by_c \in \Rp^T$ from $\hat{F}_{\mcy,c}$, 
\ben
    \i With probability $1-\dy$: allocate zero resources, i.e. $y_{c,t} = 0$ for all $t \in \mct$.
    \i With probability $\dy$: Select a type $t\in\mct$ with probability $p_t$. Draw a random $U \sim \text{Unif}[0,1]$. Then allocate $y_{c,t} = \frac{2Y_t v_c}{\dx p_t}\cdot U$ and allocate $y_{c,t'} = 0$ for all $t' \neq t$.
\een
The strategy $\hat{F}_{\mcy,c}$ either allocates zero resources of all types to contest $c$, or allocates a positive amount of a single type of resource. This strategy allocates $\E_{\by_c\sim \hat{F}_{\mcy,c}}[y_{c,t}] = v_c Y_t$ resources of type $t$ to contest $c$ in expectation, and thus it is budget-feasible.

\subsection{Proof of Theorem \ref{thm:WL}}

We provide a series of Lemmas in order to establish Theorem \ref{thm:WL}. The following Lemma establishes an upper bound on player $\mcx$ when player $\mcy$ uses any $\hat{F}_\mcy \in \F_\BS(\bY)$.

\begin{lemma}\label{lem:UB}
    Consider $\ML(\bX,\bY;W_\WL,\bv)$. Suppose $\hat{F}_\mcy \in \F_\BS(\bY)$ is a strategy of the form \eqref{eq:CAS_Y}. Then it holds that for any $F_\mcx \in \F(\bX)$,
    \be
        \pi_\mcx(F_\mcx,\hat{F}_\mcy) \leq \UB(\dy,\bs{p})
    \ee
    where $\dy \in [0,1]$ and $\bs{p} \in \Delta(\mct)$ are the parameters of $\hat{F}_\mcy$, and
    \be
    	\UB(\dy,\bs{p}) \triangleq 1-\dy + \frac{\dy^2}{2} \cdot \left(  \sum_{t \in \mct} p_t^2 \frac{X_t}{Y_t} \right).
    \ee
\end{lemma}
\begin{proof}
	For any $F_\mcx \in \F(X)$, we have that the payoff in contest $c$ is
	\be\label{eq:UB_calculation}
		\ba
			&\pi_{\mcx,c}(F_{\mcx,c},\hat{F}_{\mcy,c}) :=\E_{F_{\mcx,c}, \hat{F}_{\mcy,c}}\left[ \mathds{1}_{\{x_{c,t} \geq y_{c,t}, \forall t \in \mct \}} \right] \\
			&= \E_{F_{\mcx,c}}\left[ \hat{F}_{\mcy,c}(x) \right] \\
			&= 1-\dy + \dy \sum_{t \in \mct} p_t \cdot \E_{F_{\mcx,c}}\left[ \min\left\{ \frac{\dy p_t}{2 v_c Y_t }\cdot x_{c,t}, 1 \right\} \right] \\
			&\leq 1-\dy + \dy \sum_{t \in \mct} p_t \cdot \E_{F_{\mcx,c}}\left[ \frac{\dy p_t}{2 v_c Y_t }\cdot x_{c,t} \right] \\
			&= 1-\dy + \frac{\dy^2}{2} \cdot \left(  \sum_{t \in \mct} p_t^2 \frac{X_{c,t}}{v_c Y_t} \right) \\
		\ea
	\ee
	
	Here, $X_{c,t} := \E_{F_{\mcx,c}}[x_{c,t}]$. The inequality is due to the fact that $\min\{n_1,n_2\} \leq n_1$ for any two numbers $n_1,n_2$. The total payoff can then be written
	\be
		\ba
			&\pi_\mcx(F_\mcx,\hat{F}_\mcy) = \sum_{c\in\mcc} v_c \cdot \pi_{\mcx,c}(F_{\mcx,c},\hat{F}_{\mcy,c}) \\
			&\leq 1-\dy +\frac{\dy^2}{2}  \sum_{c\in\mcc} v_c \left(  \sum_{t \in \mct} p_t^2 \frac{X_{c,t}}{v_c Y_t} \right) \\
			&= 1-\dy +\frac{\dy^2}{2}  \sum_{t \in \mct} p_t^2\frac{1}{Y_t} \sum_{c\in\mcc} X_{c,t} \\
			&\leq 1 - \dy + \frac{\dy^2}{2}  \sum_{t \in \mct}   p_t^2 \frac{X_t}{Y_t}.
		\ea
	\ee
	The last inequality follows from \eqref{eq:ML_admissible}.
\end{proof}

The next Lemma provides the smallest upper bound that player $\mcy$ can impose by using strategies $\hat{F}_\mcy$ of the form \eqref{eq:CAS_Y}.
\begin{lemma}\label{lem:min_UB}
	Consider $\ML(\bX,\bY;W_\WL,\bv)$. Then,
	\be\label{eq:UB_OPT}
		\min_{\substack{\dy \in [0,1] \\ \bs{p} \in \Delta(\mct) } } \UB(\dy,\bs{p}) = L(\alpha(\bX,\bY)).
	\ee
	The strategy $F_\mcy^* \in \F(\bY)$ that imposes the bound upper bound  $L(\alpha(\bX,\bY))$ is given as follows. 
	If $\alpha(\bX,\bY) \leq 1$, then
    	\be
       		\ba
	            F_{\mcy,c}^*(\bs{u}) &= 1-\alpha(\bX,\bY) \\
	            &\quad + \alpha(\bX,\bY) \sum_{t \in \mct} \frac{Y_t/X_t}{\alpha(\bX,\bY)} \cdot \min\left\{ \frac{u_t}{2 v_c X_t }, 1 \right\}. \\
        		\ea
    	\ee
    	for any $c \in \mcc$ and $\bu\in\Rp^T$. If $\alpha(\bX,\bY) > 1$, then
	\be
	        \ba
	            F_{\mcy,c}^*(\bs{u}) &= \sum_{t \in \mct} \frac{Y_t/X_t}{\alpha(\bX,\bY)} \cdot \min\left\{ \frac{ u_t}{2 v_c X_t \alpha(\bX,\bY)}, 1 \right\}. \\
	        \ea
	\ee
	for any $c \in \mcc$ and $\bu\in\Rp^T$.
\end{lemma}
\begin{proof}
	The optimization problem on the left-hand side of \eqref{eq:UB_OPT} can be sequenced as
	\be
	        	\UB^* = \min_{\bs{p}\in\Delta(\mct)} \min_{\dy\in[0,1]} \UB(\dy,\bs{p}).
	\ee
	Observing that $\UB(\dy,\bs{p})$ is a convex quadratic function in $\dy \in [0,1]$, the minimizing $\dy^*$ for the inner minimization must satisfy
    \be
        \dy^* = \min\left\{\frac{1}{g(\bs{p})}, 1 \right\}
    \ee
    for any fixed $\bs{p}$, where $g(\bs{p}) := \sum_{t \in \mct} p_t^2 \frac{X_t}{Y_t}$. Then, the optimization problem can be written as
    \be\label{eq:reduced_OPT}
    	\ba
    		&\min_{\bs{p}\in\Delta(\mct)} \UB(\dy^*,\bs{p}) = \\
        		&\min_{\bs{p}\in\Delta(\mct)} \left( 1 - \min\left\{\frac{1}{g(\bs{p})}, 1 \right\} + \frac{1}{2} \min\left\{\frac{1}{g(\bs{p})}, 1 \right\}^2 g(\bs{p}) \right) \\
        	\ea
    \ee
    We now observe that $g(p)$ is a convex function on $\bs{p} \in \Delta(\mct)$, and attains its minimum at the point $\bs{p}^*$ whose entries are
    \be
        p^*_t = \frac{Y_t/X_t}{\alpha}, \ t\in\mct,
    \ee
    where $\alpha = \sum_{t'\in\mct} \frac{Y_{t'}}{X_{t'}}$ as defined in Theorem \ref{thm:WL}. We then have $g(\bs{p}^*) = 1/\alpha$.

    \noindent\textbf{Case 1:} $\alpha \leq 1$. It holds that $g(\bs{p}^*) \geq 1$, and thus $g(\bs{p}) \geq 1$ for all $\bs{p} \in \Delta(\mct)$. The optimization problem \eqref{eq:reduced_OPT} is
    \be
        \min_{\bs{p}\in\Delta(\mct)} \UB(\dy^*,\bs{p}) =\min_{\bs{p}\in\Delta(\mct)} \left( 1 - \frac{1}{2 g(\bs{p})}  \right)
    \ee
    The minimizer  is identical to the minimizer of $g(\bs{p})$, which is just $\bs{p}^*$.  Therefore, we have $U_\mcx^* = 1 - \frac{\alpha}{2}$, with $\dy^* = \alpha$.
    
    \noindent\textbf{Case 2:} $\alpha > 1$. It holds that $g(\bs{p}^*) < 1$. We can write \eqref{eq:reduced_OPT} as
    \be
        \min_{\bs{p}\in\Delta(\mct)} \UB(\dy^*,\bs{p}) =
        \begin{cases}
            \frac{g(\bs{p})}{2}, &\text{if } g(\bs{p}) < 1 \\
            1 - \frac{1}{2g(\bs{p})}, &\text{if } g(\bs{p}) \geq 1 \\
        \end{cases}
    \ee
    Observe that $\bs{p}^*$ is the minimizer in the region $g(\bs{p}) < 1$. We can further deduce that $\bs{p}^*$ also minimizes $\UB(\dy^*,\bs{p})$ over the entire simplex $\Delta(\mct)$, since $\UB(\dy^*,\bs{p}) \geq \frac{1}{2} \geq \frac{g(\bs{p}^*)}{2}$ for any $p \in \Delta(\mct)$ such that $g(p) \geq 1$. Thus, we have $\UB^* = \frac{1}{2\alpha}$, with $\dy^* = 1$.
   
   We then obtain $\UB^* =L(\alpha(\bX,\bY))$, and the form of the strategy $F_\mcy^*$ follows from \eqref{eq:CAS_Y}.

\end{proof}

We note that the smallest upper bound from Lemma \ref{lem:UB} coincides with the reported equilibrium payoff in Theorem \ref{thm:WL}. Our next Lemma establishes lower bounds.

\begin{lemma}\label{lem:LB}
    Consider $\ML(\bX,\bY;W_\WL,\bv)$. Suppose $\hat{F}_\mcx \in \F(\bX)$ is a strategy of the form \eqref{eq:CAS_X}. Then it holds that for any $F_\mcy \in \F(\bY)$,
    \be
        \pi_\mcx(\hat{F}_\mcx,F_\mcy) \geq \LB(\dx)
    \ee
    where $\dx \in [0,1]$ is the parameter of $\hat{F}_\mcx$, and
    \be
    	\LB(\dx) \triangleq \dx\left(1 - \frac{\dx}{2}\alpha(\bX,\bY) \right)
    \ee
\end{lemma}

\begin{proof}

   	Suppose $\mcx$ utilizes a strategy $\hat{F}_\mcx$ of the form \eqref{eq:CAS_X}. For any $F_\mcy \in \F(Y)$, the payoff to $\mcy$ in contest $c$ is 
    	\be
    		\ba
	    		 &\pi_{\mcy,c}(\hat{F}_{\mcx,c},F_{\mcy,c}) :=  1 - \E_{\hat{F}_{\mcx,c}, F_{\mcy,c}}\left[ \mathds{1}_{\{x_{c,t} \geq y_{c,t}, \text{ for all } t \in \mct \}} \right] \\
	    		 &= 1 - \E_{\hat{F}_{\mcy,c} } \left[ \P_{\hat{F}_{\mcx,c}}\left( x_{c,t} \geq y_{c,t}, \text{ for all } t \in \mct \right) \right] \\
	 		&= 1 - \E_{\hat{F}_{\mcy,c} } \left[1 - \sum_{\mcs \subseteq \mct, |\mcs| \geq 1} (-1)^{|\mcs|+1}  \hat{F}_{\mcx,c,\mcs}(\by_c) \right]
    		\ea
   	 \ee
    	The equality above is due to the Inclusion-Exclusion principle, and $F_{\mcx,c,\mcs}(\by_c) = \lim_{y_{c,d} \rightarrow \infty, \forall d\notin\mcs}$
    	\be
	    	 \hat{F}_{\mcx,c,\mcs}(\by_c) := \lim_{y_{c,d} \rightarrow \infty, \forall d\notin\mcs} \hat{F}_{\mcx,c}(\by_c)
 	\ee
 	is the marginal distribution of $ \hat{F}_{\mcx,c}$ with respect to variables $\mcs \subseteq \mct$. Denoting $m_{\mcx,c,t}(y_t) := \min\left\{ \frac{\dx}{2 v_c X_t }\cdot y_t, 1 \right\}$, we can write  from \eqref{eq:CAS_X}:
	\be
		\ba
	    	 	\hat{F}_{\mcx,c,\mcs}(\by_c) &= (1 - \dx) + \dx \min\left\{ m_{\mcx,c,t}(y_t) \right\}_{t\in\mcs} \\
	    	 \ea
 	\ee
    	Continuing, we have
    	\be
		\ba
	    	 	&\pi_{\mcy,c}(\hat{F}_{\mcx,c},F_{\mcy,c}) \\
	    	 	&= 1 - \E_{\hat{F}_{\mcy,c} } \left[ 1- (1-\dx)\sum_{\mcs \subseteq \mct, |\mcs| \geq 1} (-1)^{|\mcs|+1}\right. \\
	    	 	&\left. \quad-\dx  \sum_{\mcs \subseteq \mct, |\mcs| \geq 1} (-1)^{|\mcs|+1}  \min\left\{ m_{\mcx,c,t}(y_t) \right\}_{t\in\mcs} \right] \\
	    	 	&= \E_{\hat{F}_{\mcy,c} } \left[ (1-\dx) \right. \\
	    	 	&\left.\quad+ \dx \sum_{\mcs \subseteq \mct, |\mcs| \geq 1} (-1)^{|\mcs|+1}  \min\left\{ m_{\mcx,c,t}(y_t) \right\}_{t\in\mcs} \right] \\
	    	 	&=  \E_{\hat{F}_{\mcy,c} } \left[ (1-\dx) + \dx \max\left\{ m_{\mcx,c,t}(y_t) \right\}_{t\in\mct} \right] \\
	    	 	&\leq  \E_{\hat{F}_{\mcy,c} } \left[ (1-\dx) + \dx \sum_{t\in\mct}  m_{\mcx,c,t}(y_t)  \right]
	    	 \ea
 	\ee
    	The second equality is due to the Binomial Theorem: $\sum_{k=1}^T \binom{T}{k} (-1)^{k+1} = 1$. The third equality follows from the technical result  Lemma \ref{lem:inc_exc} found in the Appendix. The inequality in the last line above follows from the fact that $\max_{t} \{n_t\} \leq \sum_t n_t$ for any collection of numbers $n_t$. Moreover, since $m_{\mcx,c,t}(y_t) \leq \frac{\dx}{2v_c X_t} y_t$, we can write
    	\be
		\ba
	    	 	\pi_{\mcy,c}(\hat{F}_{\mcx,c},F_{\mcy,c}) &\leq  \E_{\hat{F}_{\mcy,c} } \left[ (1-\dx) + \frac{\dx^2}{2} \sum_{t\in\mct} \frac{y_t}{v_c X_t}  \right] \\
	    	 	&= (1-\dx) + \frac{\dx^2}{2} \sum_{t\in\mct}  \frac{\E_{\hat{F}_{\mcy,c} }[y_{c,t}]}{v_c X_t} \\
	    	 	&\leq (1-\dx) + \frac{\dx^2}{2} \sum_{t\in\mct}  \frac{Y_t}{X_t} \\
	    	 \ea
 	\ee
	The final inequality follows from $F_\mcy \in \F(Y)$ \eqref{eq:ML_admissible}. Summing over all contests to get the total payoff $\pi_{\mcy}(\hat{F}_{\mcx},F_{\mcy}) = \sum_{c\in\mcc} v_c \cdot \pi_{\mcy,c}(\hat{F}_{\mcx,c},F_{\mcy,c})$, we obtain the bound
	\be
		\pi_{\mcy}(\hat{F}_{\mcx},F_{\mcy}) \leq (1-\dx) + \frac{\dx^2}{2} \sum_{t\in\mct}  \frac{Y_t}{X_t},
	\ee
	as the contest values are normalized. We can equivalently express this as a lower bound on player $\mcx$'s payoff,
	\be
		\pi_{\mcx}(\hat{F}_{\mcx},F_{\mcy}) \geq \dx\left(1 - \frac{\dx}{2} \sum_{t\in\mct}  \frac{Y_t}{X_t} \right).
	\ee
\end{proof}

The next Lemma provides the greatest lower bound that player $\mcx$ can ensure by using strategies $\hat{F}_\mcx$ of the form \eqref{eq:CAS_X}.
\begin{lemma}\label{lem:max_LB}
	Consider $\ML(\bX,\bY;W_\WL,\bv)$. Then,
	\be\label{eq:LB_OPT}
		\max_{\dx \in [0,1] } \LB(\dx) = L(\alpha(\bX,\bY)).
	\ee
	The strategy $F_\mcx^* \in \F(\bX)$ that ensures the payoff  $L(\alpha(\bX,\bY))$ is given as follows. 
	If $\alpha(\bX,\bY) \leq 1$, then
    	\be
       		\ba
	            F_{\mcx,c}^*(\bs{u}) &= \min\left\{ \min\left\{ \frac{u_t}{2 v_c X_t }, 1 \right\} \right\}_{t \in \mct} \\
        		\ea
    	\ee
    	for any $c \in \mcc$ and $\bu\in\Rp^T$. If $\alpha(\bX,\bY) > 1$, then
	\be
	        \ba
	            F_{\mcx,c}^*(\bs{u}) &= 1 - \frac{1}{\alpha(\bX,\bY)} \\
	            & +  \frac{1}{\alpha(\bX,\bY)}\cdot \min\left\{ \min\left\{ \frac{u_t}{2 v_c X_t \alpha(\bX,\bY) }, 1 \right\} \right\}_{t \in \mct} \\
	        \ea
	\ee
	for any $c \in \mcc$ and $\bu\in\Rp^T$.
\end{lemma}
\begin{proof}
	The maximum lower bound attainable through the choice of $\dx$ is determined by solving the optimization problem
	\be
	    \LB^* = \max_{\dx \in [0,1]} \LB(\dx),
	\ee
	which has the solution
	\be
	    \dx^* = \min\left\{ \frac{1}{\alpha}, 1
	    \right\}.
	\ee
	We then obtain $\LB^* =L(\alpha(\bX,\bY))$, and the form of the strategy $F_\mcx^*$ follows from \eqref{eq:CAS_X}.
\end{proof}

We are now ready to establish the main result, Theorem \ref{thm:WL}.

\begin{proof}[Proof of Theorem \ref{thm:WL}]
	From Lemmas \ref{lem:min_UB} and \ref{lem:max_LB}, we obtain
	\be
		\pi_\mcx(F_\mcx,F_\mcy^*) \leq L(\alpha) \leq \pi_\mcx(F_\mcx^*,F_\mcy) 
	\ee
	for all $F_\mcx \in \F(\bX)$ and $F_\mcy \in \F(\bY)$. They also imply that $\pi_\mcx(F_\mcx^*,F_\mcy^*) = L(\alpha)$. Thus, the strategy profile $(F_\mcx^*,F_\mcy^*)$ is an equilibrium (Definition \ref{def:equil}).
\end{proof}

\section{Conclusion}

The allocation of heterogeneous resources is increasingly becoming a significant aspect for system planners to consider, given the complexity and diversity of objectives that need to be accomplished. In this manuscript, we considered the allocation of multiple and heterogeneous resource types for a strategic adversarial setting known as General Lotto games. We extended the classic and pre-dominantly single-resource General Lotto games considered in the literature to scenarios where players allocate multiple types of resources. Here, the specification of winning rules for each individual contest is the novel modeling consideration in our framework. We derived equilibrium characterizations in the cases where the winning rules are based on weakest-link objectives, and when they are based on a weighted linear combination of the allocated resource types. We then analyzed how  investment decisions should be made when it is costly to utilize each resource type.

\bibliographystyle{IEEEtran}
\bibliography{sources}

 \appendix

\subsection{Technical Result for Theorem \ref{thm:WL}}

The following result is utilized in the proof of Theorem \ref{thm:WL}, specifically in Lemma \ref{lem:LB}.

\begin{lemma}\label{lem:inc_exc}
    Suppose $(z_t)_{t\in \mct} \in \Rp^T$ is a vector of non-negative numbers, where we denote $T = |\mct|$. Then,
    \begin{align*}
        \sum_{S\subseteq \mct, |S|>0}(-1)^{|S|+1}\min_{t\in S} z_t=\max_{t\in T}z_t.
    \end{align*}
\end{lemma}
\begin{proof}
    We prove the lemma by induction on $T$, the size of the index set $\mct$. The case $T=1$ trivially holds. Therefore, suppose that the statement is correct for arbitrary $T$. We want to prove that it is correct for $T+1$. Let $m=\min\{z_T,z_{T+1}\}$ and $M=\max\{z_T,z_{T+1}\}$. We have
    \begin{align*}
        \sum_{S\subseteq [T+1]}(-1)^{|S|+1}\min_{t\in S} z_t&=\sum_{\substack{S\subseteq [T+1]\\m\in S, M\not\in S}}(-1)^{|S|+1}\min_{t\in S} z_t\cr
        &+\sum_{\substack{S\subseteq [T+1]\\m\not\in S, M\in S}}(-1)^{|S|+1}\min_{t\in S} z_t\cr
        &+\sum_{\substack{S\subseteq [T+1]\\m\in S, M\in S}}(-1)^{|S|+1}\min_{t\in S} z_t.
    \end{align*}
    Note that for every $S\subseteq [T+1]$ such that $m\in S, M\not\in S$, we have the set $S\cup\{M\}\subseteq [T+1]$. Also, for every $S\subseteq [T+1]$ such that $m\in S, M\in S$, we have the set $S\setminus\{M\}\subseteq [T+1]$. Therefore, since 
    \begin{align*}
        (-1)^{|S|+1}\min_{t\in S} z_t+(-1)^{|S|+2}\min_{t\in S\cup\{M\}} z_t=0
    \end{align*}
    where $m\in S, M\not\in S$,
    we have
    \begin{align*}
        \sum_{\substack{S\subseteq [T+1]\\m\in S, M\not\in S}}&(-1)^{|S|+1}\min_{t\in S} z_t\cr
        &+\sum_{\substack{S\subseteq [T+1]\\m\in S, M\in S}}(-1)^{|S|+1}\min_{t\in S} z_t=0.
    \end{align*}
    Thus, we obtain
    \begin{align*}
        \sum_{S\subseteq [T+1]}(-1)^{|S|+1}\min_{t\in S} z_t&=
        \sum_{\substack{S\subseteq [T+1]\\m\not\in S, M\in S}}(-1)^{|S|+1}\min_{t\in S} z_t\cr
        &=\sum_{\substack{S\subseteq [T+1]\setminus\{m\}}}(-1)^{|S|+1}\min_{t\in S} z_t\cr
        &=\max\{z_1,\ldots,z_{T-1},M\}\cr
        &=\max_{t\in[T+1]}z_t,
    \end{align*}
    where the third equality follows from the induction assumption on sets of size $T$.
\end{proof}

\subsection{Proof of Lemma \ref{lem:sunk}}

Consider the game $\MLM(M_\mcx, M_\mcy,\bs{\costx},\bs{\costy},W_\WL)$. The best-response problem for player $\mcx$ is
\be
	\ba
		\max_{\bX} L(\alpha(\bX,\bY)) \text{ s.t. }
		\begin{cases}
			&\sum_{t\in\mct} \costx_t X_t \leq M_\mcx \\
			&X_t \geq 0, \ \forall t \in \mct \\
		\end{cases}
	\ea
\ee
with $\bY$ fixed. The necessary condition is
\be
	L'(\alpha(\bX,\bY))\cdot \left(\frac{Y_t}{X_t^2} \right) + \lambda \costx_t - \lambda_t = 0, \ \forall t \in \mct.
\ee
where $\lambda \geq 0$ is the multiplier associated with the monetary budget constraint, and $\lambda_t$ is the multiplier associated with the constraint $-X_t \leq 0$. We first observe that any optimal response must have $X_t > 0$, for otherwise the payoff is zero. Moreover, the payoff is strictly increasing in $X_t$ for any $t\in\mct$, so an optimal response must utilize the entire monetary budget $M_\mcx$. Thus, $\lambda_t = 0$ for all $t \in \mct$ and $\lambda > 0$. We  then get
\be
	X_t = \sqrt{\frac{-L'(\alpha(\bX,\bY)) Y_t}{\lambda \cdot \costx_t}}
\ee
We can then write the budget constraint as $M_\mcx =  \sqrt{\frac{-L'(\alpha(\bX,\bY))}{\lambda}} \sum_{t\in\mct} \sqrt{\frac{Y_t}{\costx_t}}$. The multiplier is evaluated to be 
\be
	\lambda = \frac{-L'(\alpha(\bX,\bY))}{M_\mcx^2} \sum_{t\in\mct} \sqrt{\frac{Y_t}{\costx_t}}.
\ee
Using this expression, the optimal allocation can then be calculated to be
\be\label{eq:Xstar}
	X_t^* = \frac{M_\mcx}{\costx_t} \frac{\sqrt{\costx_t Y_t}}{ \sum_{s\in\mct} \sqrt{\costx_s Y_s} }.
\ee

To find the equilibrium, we can now consider the best-response problem for player $\mcy$ as
\be
	\ba
		\min_{\bY} L(\alpha(\bX^*,\bY)) \text{ s.t. }
		\begin{cases}
			&\sum_{t\in\mct} \costy_t Y_t \leq M_\mcy \\
			&Y_t \geq 0, \ \forall t \in \mct \\
		\end{cases}.
	\ea
\ee
where $\bX^*$ is given in \eqref{eq:Xstar}. The necessary condition is
\be
	L'(\alpha(\bX^*,\bY))\cdot \left(\frac{1}{X_t^*} \right) + \lambda \costy_t - \lambda_t = 0, \ \forall t \in \mct.
\ee
Using a change of variable $Z_t := \sqrt{\costx_t Y_t}$, we obtain
\be\label{eq:zt_lambda}
	Z_t = \frac{-L'(\alpha(\bX^*,\bY)) }{M_\mcx \lambda } \frac{\costx_t}{\costy_t}  \sum_{s\in\mct} Z_s
\ee
The budget constraint requires $\sum_{t\in\mct} \costy_t \frac{Z_t^2}{\costx_t} = M_\mcy$. Using the expression \eqref{eq:zt_lambda}, we can evaluate the multiplier to be
\be
	\lambda = \frac{-L'(\alpha(\bX^*,\bY))}{M_\mcx \sqrt{M_\mcy}} \left( \sum_{s\in\mct} Z_s \right) \sqrt{r }
\ee
where $r = \sum_{s\in\mct} \frac{\costx_s}{\costy_s}$. We then obtain
\be
	Z_t = \sqrt{\frac{M_\mcy}{r}} \cdot \frac{\costx_t}{\costy_t}.
\ee
The optimal allocation $Y_t^*$ is recovered via 
\be
	Y_t^* = \frac{Z_t^2}{\costx_t} =  \frac{M_\mcy}{\costy_t} \frac{\costx_t/\costy_t}{r}.
\ee
The optimal allocation of player $\mcx$ can then be written as
\be
	X_t^* = \frac{M_\mcx}{\costx_t} \frac{\costx_t/\costy_t}{r}.
\ee

\begin{IEEEbiography}[{\includegraphics[width=1in,height=1.25in,clip,keepaspectratio]{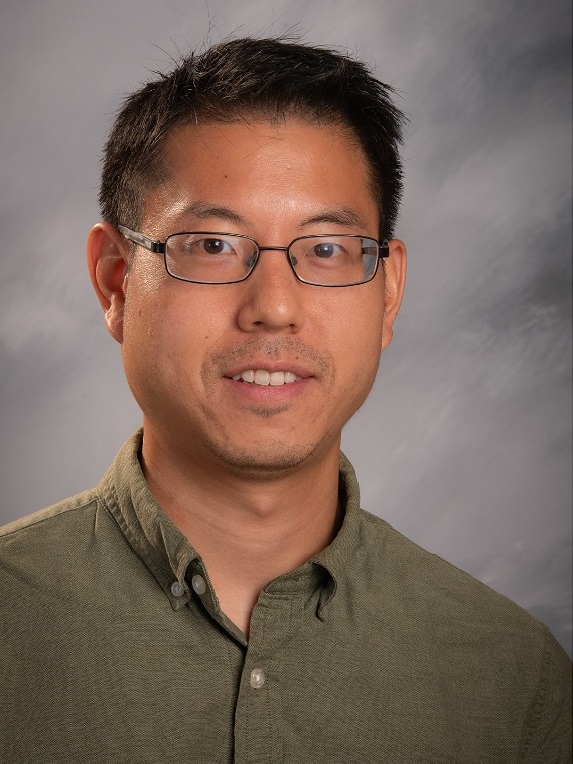}}] {Keith Paarporn}
  is an Assistant Professor in the Department of Computer Science at the University of Colorado, Colorado Springs. He received a B.S. in Electrical Engineering from the University of Maryland, College Park in 2013, an M.S. in Electrical and Computer Engineering from the Georgia Institute of Technology in 2016, and a Ph.D. in Electrical and Computer Engineering from the Georgia Institute of Technology in 2018. From 2018 to 2022, he was a postdoctoral scholar in the Electrical and Computer Engineering Department at the University of California, Santa Barbara. His research interests include game theory, control theory, and their applications to multi-agent systems and security. 
\end{IEEEbiography}

\begin{IEEEbiography}[{\includegraphics[width=1in,height=1.25in,clip,keepaspectratio]{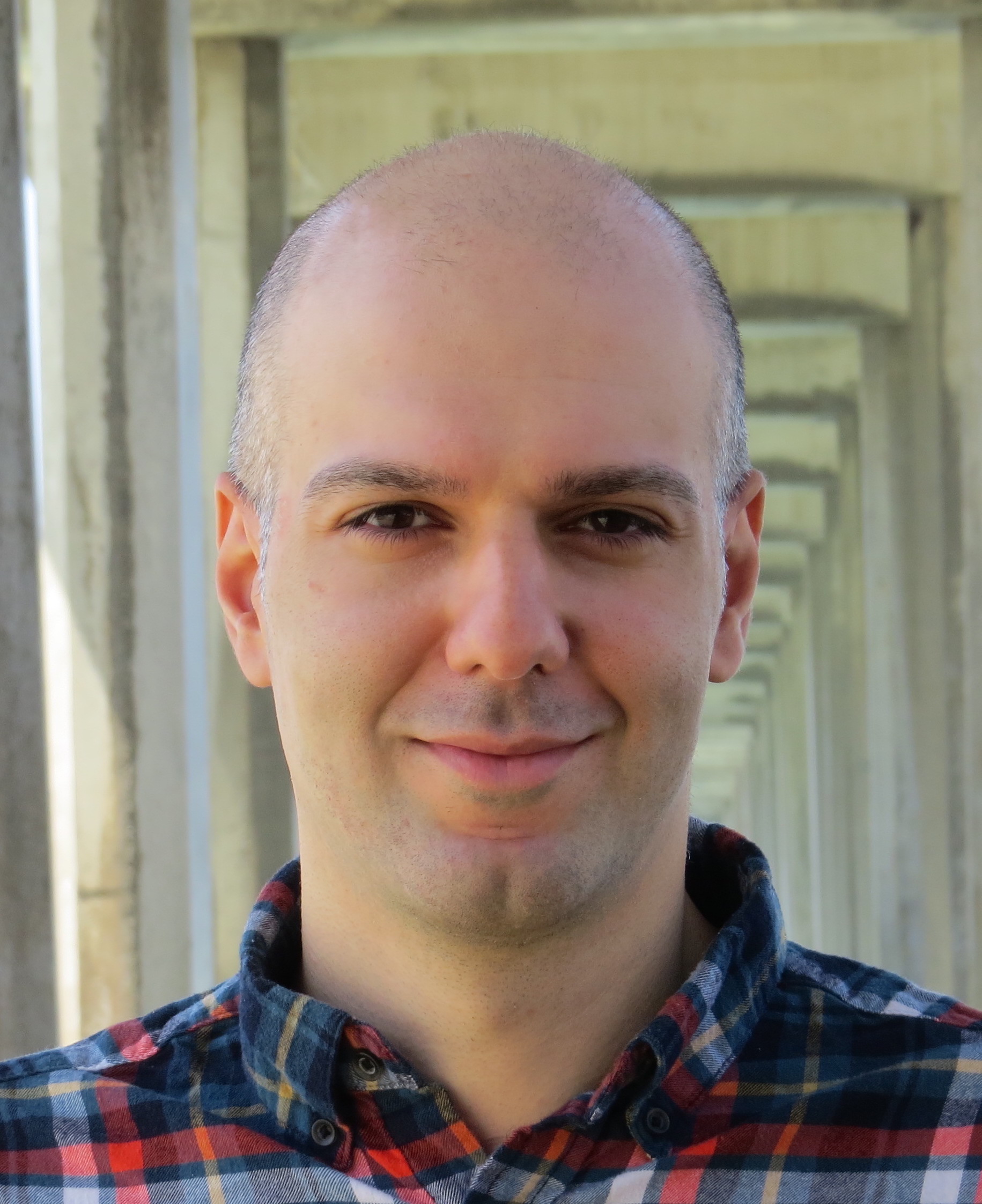}}] {Adel Aghajan}  received a B.Sc., and M.Sc. degrees in Electrical Engineering
 from Isfahan University of Technology, Isfahan, Iran, in 2009, and 2012,
respectively, and  a Ph.D.\ degree in Electrical Engineering at University of California, San Diego, in 2021. He is currently a Postdoctoral Researcher with the
Department of Electrical and Computer Engineering at the University of California, San Diego. His research interests include game theory, distributed computation
and optimization
, random dynamics, and information theory. 
    
\end{IEEEbiography}

\begin{IEEEbiography}[{\includegraphics[width=1in,height=1.25in,clip,keepaspectratio]{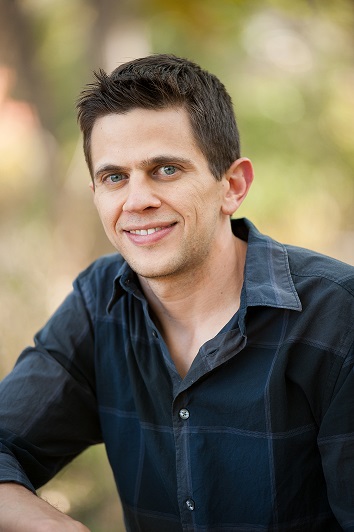}}] {Jason Marden}
    is a Professor in the Department of Electrical and Computer Engineering at
    the University of California, Santa Barbara. Jason received a BS in Mechanical Engineering in 2001 from
    UCLA, and a PhD in Mechanical Engineering in
    2007, also from UCLA.
    After graduating from UCLA, he served as a junior
    fellow in the Social and Information Sciences Laboratory at the California Institute of Technology until
    2010 when he joined the University of Colorado. In 2015, Jason joined the Department of Electrical and Computer Engineering at UCSB.  Jason is a recipient of the NSF Career Award (2014), the ONR Young Investigator Award (2015), the AFOSR Young Investigator Award (2012), the American Automatic Control Council Donald P. Eckman Award (2012), the SIAG/CST Best SICON Paper Prize (2015), and was named an IEEE Fellow (2023). Jason’s research interests focus on game theoretic methods for the control of distributed multiagent systems.
\end{IEEEbiography}

\end{document}